\documentclass[a4paper,10pt]{article}
\pdfoutput=1
\usepackage[bookmarks=false]{hyperref}
\usepackage[numbers,sort&compress]{natbib}
\usepackage[margin=1in]{geometry}

\usepackage{etex}

\usepackage{amsmath}
\usepackage{amsfonts}
\usepackage{amssymb}

\usepackage{amsthm}

\usepackage[noadjust]{cite}

\usepackage{algorithmic}
\usepackage{url}
\usepackage{booktabs}
\usepackage{tikz,pgfplots}
\pgfplotsset{compat=newest}
\usetikzlibrary{external}
\tikzsetexternalprefix{tikzfigures/}
\usepgfplotslibrary{fillbetween}
\usepackage{url}

\usepackage[ruled, vlined, nofillcomment]{algorithm2e}
\usepackage{algorithmic}

\SetCommentSty{mycommfont}
\SetAlgoCaptionSeparator{:}

\usepackage{framed}

\newcommand{\Transp}{\mathsf{T}}

\def\d{\mathrm{d}}

\newcommand{\ie}{i.e.\ }
\newcommand{\eg}{e.g.\ }

\usepackage{empheq}

\newlength\figureheight 
\newlength\figurewidth 

\newcommand{\mytilde}{\raise.17ex\hbox{$\scriptstyle\mathtt{\sim}$}}

\theoremstyle{plain}
\newtheorem{theorem}{Theorem}
\newtheorem{remark}{Remark}

\newtheorem{definition}{Definition}

\graphicspath{{./figures/}}

\providecommand{\keywords}[1]{\textbf{\textit{Keywords: }} #1}

\begin{document}
% !TEX spellcheck = en-UK

\newcommand{\coverTitle}{Observability of the relative motion from inertial data \\ in kinematic chains}
\newcommand{\coverYear}{2022}

\newcommand{\coverAuthors}{Manon~Kok$^{\star}$, Karsten~Eckhoff$^{\star \star}$, Ive Weygers$^{\dagger}$ and Thomas Seel$^{\star \star}$ \\ \vspace{3mm}
\small{$^\star$Delft Center for Systems and Control, Delft University of Technology, the Netherlands. \\E-mail: m.kok-1@tudelft.nl} \\
\small{$^{\star \star}$Control Systems Group, Technische Universit\"at Berlin, Germany. \\ Emails: \{karsten.eckhoff@campus.tu-berlin.de, thomas.seel@tu-berlin.de\}} \\
\small{$^\dagger$Department of Rehabilitation Sciences, KU Leuven Campus Bruges, Belgium. \\E-mail: ive.weygers@kuleuven.be} 
}

\begin{titlepage}
\begin{center}
%
%% 
%{\large \em Technical report}

\vspace*{2.5cm}
%
%% TITLE
{\Huge \bfseries \coverTitle  \\[0.4cm]}

%
%% AUTHORS
{\Large \coverAuthors \\[1.5cm]}

\renewcommand\labelitemi{\color{red}\large$\bullet$}
\begin{itemize}
\item {\Large \textbf{Please cite this version:}} \\[0.4cm]
\normalsize
Manon Kok, Karsten Eckhoff, Ive Weygers and Thomas Seel, ``Observability of the relative motion from inertial data in kinematic chains", Control Engineering Practice, Volume 125, 105206, August 2022. 

DOI: https://doi.org/10.1016/j.conengprac.2022.105206
\end{itemize}

\end{center}

\vspace{1cm}

\begin{abstract}
    Real-time motion tracking of kinematic chains is a key prerequisite in the control of, e.g., robotic actuators and autonomous vehicles and also has numerous biomechanical applications. In recent years, it has been shown that, by placing inertial sensors on segments that are connected by rotational joints, the motion of that kinematic chain can be tracked accurately. These methods specifically avoid using magnetometer measurements, which are known to be unreliable since the magnetic field at the different sensor locations is typically different. They rely on the assumption that the motion of the kinematic chain is sufficiently rich to assure observability of the relative pose. However, a formal investigation of this crucial requirement has not yet been presented, and no specific conditions for observability have so far been given. In this work, we present an observability analysis and show that the relative pose of the body segments is indeed observable under a very mild condition on the motion. We support our results by simulation studies, in which we employ a state estimator that neither uses magnetometer measurements nor additional sensors and does not impose assumptions on the accelerometer to measure only the direction of gravity, nor on the range of motion or degrees of freedom of the joints. We investigate the effect of the amount of excitation and of stationary periods in the data on the accuracy of the estimates. We then use experimental data from two mechanical joints as well as from a human gait experiment to validate the observability criterion in practice and to show that small excitation levels are sufficient for obtaining accurate estimates even in the presence of time periods during which the motion is not observable.
\end{abstract}


\vfill

\end{titlepage}

\title{Observability of the relative motion from inertial data in kinematic chains}

\author{Manon~Kok$^{\star}$, Karsten~Eckhoff$^{\star \star}$, Ive Weygers$^{\dagger}$ and Thomas Seel$^{\star \star}$ \\ \vspace{3mm}
\small{$^\star$Delft Center for Systems and Control, Delft University of Technology, the Netherlands. E-mail: m.kok-1@tudelft.nl} \\
\small{$^{\star \star}$Control Systems Group, Technische Universit\"at Berlin, Germany. Emails: \\ \{karsten.eckhoff@campus.tu-berlin.de, thomas.seel@tu-berlin.de\}} \\
\small{$^\dagger$Department of Rehabilitation Sciences, KU Leuven Campus Bruges, Belgium. E-mail: ive.weygers@kuleuven.be} 
}
\date{\empty}

\section{Introduction}
\label{sec:introduction}
In recent years, inertial measurement units (IMUs) have been used for motion tracking and control in an increasing number of mechatronic and biomechanical applications ranging from autonomous cars, miniature aerial vehicles and offshore vessels \citep{RodrigoMarco:2020,Hoffmann:2010,Wan:2018,bryneRFJ:2018} to human motion capture and feedback-controlled biomedical devices \citep{Zihajehzadeh:2015,Seel:2015}. These sensors typically contain three-dimensional gyroscopes, accelerometers and magnetometers. Accelerometers and gyroscopes, measuring the specific force and the angular velocity, respectively, are also called inertial sensors.

Our interest lies in estimation of the motion of kinematic chains consisting of segments that are connected by rotational joints, where each of the segments is equipped with an IMU. These could for example be human body segments, multilink aerial vehicles or robotic actuators, as illustrated in Figure~\ref{fig:figure1}. One obvious but restrictive way to do this is to estimate the orientation of each IMU individually. A major limitation of this approach is that it relies on the assumptions that the magnetometer approximately measures a constant local magnetic field and that the accelerometer approximately measures the gravity~\citep{kokHS:2017}. Both assumptions are often violated in practice due to the presence of ferromagnetic material or electronic devices \citep{devries2009magnetic,shu2015magicol} or due to fast motion~\citep{benallegueBC:2017}. In practice, this leads to unpredictably large estimation errors and to instability and failure of control systems that rely on these estimates. One way to overcome these limitations is to use additional sensors, see e.g.\ \citep{RodrigoMarco:2020,vigneKMMP:2018}. On the other hand, in recent years, magnetometer-free approaches have been developed for accurate and reliable estimation of the complete relative pose of all segments of a kinematic chain from only inertial measurements. Our interest lies in the latter approach.

\begin{figure}[t]
\centering
\includegraphics[width=0.6\columnwidth]{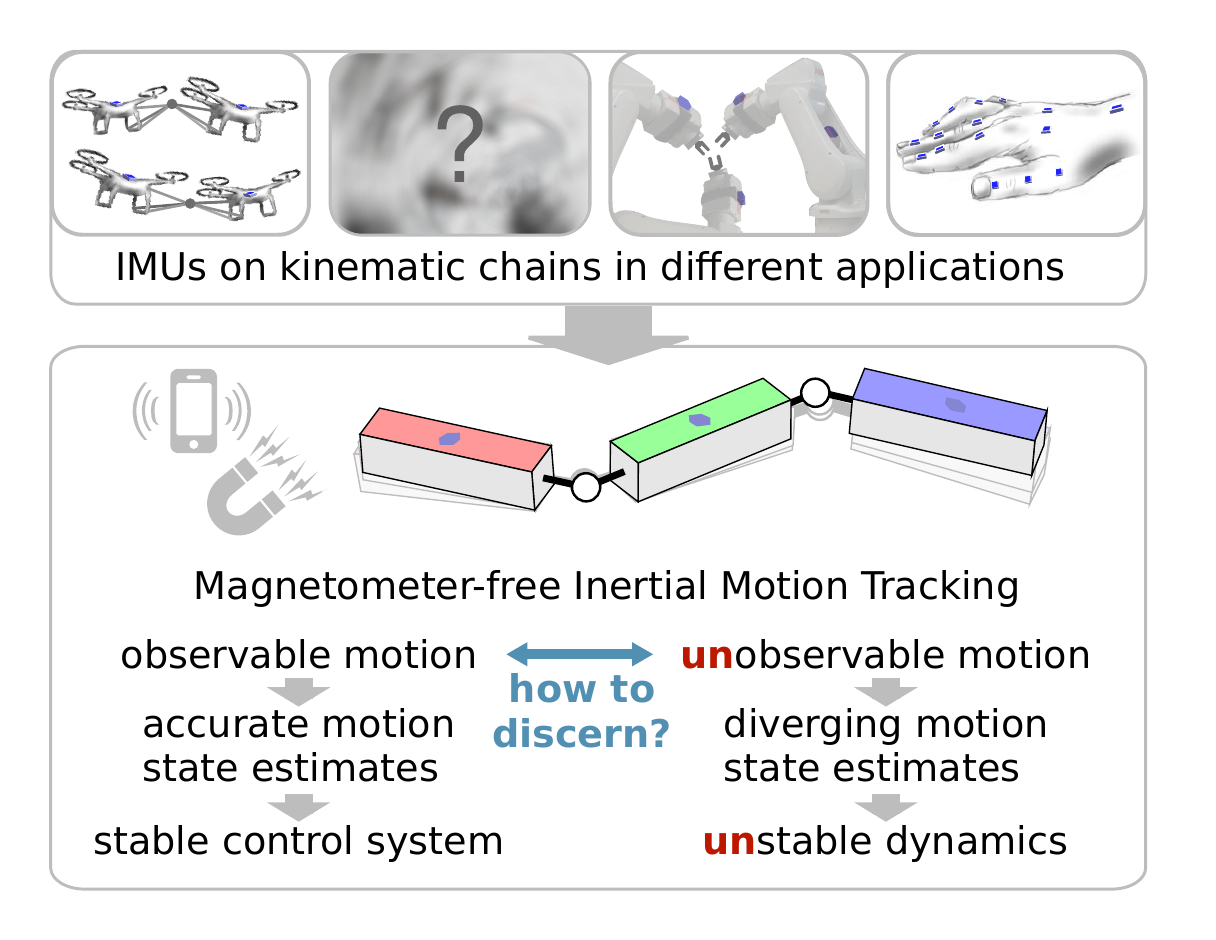}
\caption{Illustration of the main research question. Magnetometer-free inertial motion tracking of kinematic chains is desirable in many applications, including linked aerial vehicles, robotics, and human motion tracking. However, there are to date no sufficient conditions for observability of the performed motion, which puts accuracy and stability at~risk.}
\label{fig:figure1}
\end{figure}

Previous work has shown that the relative pose of the body segments can be determined entirely from inertial measurements -- without the assumption that the accelerometer only measures the gravity -- by taking the connection between the segments and the corresponding kinematic constraints into account~\citep{kokHS:2014,hol:2011,weygersKVVVHC:2020}, however only under the assumption that the motion is ``sufficiently rich''. Different types of state estimators have been used, including extended Kalman filters and optimisation-based methods. These state estimators have experimentally been shown to result in accurate estimation results with joint angle estimation errors in the range of a few degrees for biomechanical systems~\citep{weygersKVVVHC:2020}. The argument that the performed motions must be sufficiently rich to render the relative orientation observable is found throughout the mentioned literature, and it is often claimed that this is a mild assumption. However, precise statements on necessary or sufficient conditions have to date only been found for the special case of a double-hinge joint system~\citep{Eckhoff2020_IFAC}. From a practical point of view, this represents a severe limitation of magnetometer-free inertial motion tracking, since the methods must be used without knowing whether the results are accurate for the specific motion or not. 

In the present work we derive, for the first time, sufficient conditions for observability of the relative pose of the kinematic chain. This lays the foundation for crucial performance guarantees in a large range of applications. It can also be used to instruct users to perform certain movements to guarantee this performance. Since we consider arbitrary rotational joints, i.e.\ without restrictions on the range of motion, the derived sufficient conditions do not only hold in 3D joints but straightforwardly also in joints with only one or two rotation axes. 

The question under which conditions the relative motion states of a kinematic chain can be determined from magnetometer-free IMU readings and a kinematic constraint leads to observability analysis of a nonlinear dynamic system. It is known from systems and control theory that this is a non-trivial problem and that no general closed-form condition for observability of nonlinear systems exists. However, some theoretical results exist~\citep{besancon:2007,besancon:2016,hermannK:1977}, and we leverage them to derive a practically useful condition for accurate results in magnetometer-free inertial motion analysis. More specifically, we identify for which types of motion the relative motion states of the kinematic chain are not observable. We then show that even if observable motions are interrupted by short and / or infrequent periods of non-observable motions, it is possible for a state estimator to provide accurate estimation results by exploiting the kinematic constraint. While long non-observable periods inevitably lead to increasing errors in real-time estimation, the proposed criterion can be used to detect such periods and to provide crucial information about whether the estimates are reliable as well as what type of motions should be performed to increase the estimation accuracy. The contributions of the present work are five-fold: 
\begin{enumerate}
\item We show the mathematical equivalence of the two dynamic models that are most commonly used to take kinematic constraints due to the connection between body segments into account. Moreover, we propose two reduced model representations which are also mathematically equivalent but have state vectors of a smaller size.
\item We present an observability analysis that is valid for all four state space models and derive a (mild) condition under which the relative orientation can be uniquely determined from the system outputs. 
\item We analyse the estimation accuracy for both a filtering and a smoothing approach~\citep{weygersKVVVHC:2020} as a function of the type of motion and illustrate precisely in which way a motion must be sufficiently exciting to assure convergence using Monte Carlo simulations. This includes showing that a strong excitation of the ``wrong'' kind can be useless, while a small excitation of the ``right'' kind results in small estimation errors.
\item We introduce a metric that quantifies the ``amount of observability'' -- in terms of the linear independence of two vectors -- at a specific time instant and show that this metric can indeed be used to detect periods of unobservability or poor observability and can hence be used to identify periods during which the estimation results should be taken with caution.
\item We practically validate the derived conditions for two experiments. In the first, we estimate the motion of two mechanical joint systems with different degrees of freedom based on measurements from two attached IMUs. In the second, we estimate the motion of a human leg using measurements from two IMUs, placed on the thigh and the shank.
\end{enumerate}

\section{Related work}
\subsection{Magnetometer-free inertial motion tracking of kinematic chains}
In previous work, magnetometer-free approaches have been proposed that determine the \emph{relative} orientation between segments by taking the connection between the segments and the corresponding kinematic constraints into account. The practical relevance of such methods is high, since they enable accurate motion tracking in arbitrary magnetic environments and thus in a wide field of applications. In the following, we briefly summarise existing magnetometer-free approaches for relative-motion tracking of kinematic chains, and we answer two questions: Which magnetometer-free methods have been proposed for inertial motion capture, and what is known about the conditions under which these methods are known to work or fail? To estimate the motion of a kinematic chain, the most common approach is to, as we also do in this work, use a full IMU setup, which means that an IMU is placed on each segment of the kinematic chain. Methods with sparse sensor setups have been proposed but require additional assumptions on the kinematic structure and its motion \citep{Huang:2018,Marcard:2017,Eckhoff2020_IFAC} and are outside the scope of this work. 

For joints with only one degree of freedom, i.e.\ hinge joints or revolute joints, several methods have been proposed that exploit constraints of the relative orientation between the adjacent segments \citep{cooper2009inertial,Laidig:2017}. However, the kinematic constraints become singular when the hinge joint axis is vertical, and the relative heading cannot be tracked if the systems remains at the singularity. Moreover, to apply such methods, the direction of the joint axis must be known in the sensor frames of both adjacent segments, which requires suitable sensor-to-segment calibration routines \citep{seelSR:2012, olssonKSH:2020}. For joints with two degrees of freedom, e.g.\ saddle joints or the human elbow, magnetometer-free methods have been proposed that determine the relative heading by exploiting kinematic constraints of the angular rates \citep{Laidig:2019} and the orientations \citep{luinge2007ambulatory}. As in the hinge joint case, the relative heading cannot be tracked if one of the joint axes remains vertical, and the methods require identification of both joint axes' coordinates in the corresponding sensor frame \citep{LaidigMS:2017}. Finally, for joints with up to three degrees of freedom, range-of-motion constraints can be exploited on a moving-window to track the relative heading \citep{Lehmann:2020}. This approach requires knowledge of the range of motion and persistent excitation in the sense of sufficient coverage of that range. 

The aforementioned methods have in common that they apply only to joints with limited degrees of freedom or range of motion. Besides such rotational constraints, it has also been shown that one can exploit the mere fact that the segment ends that are connected by the joint cannot move apart. This information is independent of the joint's degrees of freedom and range of motion, and it can be formalised in two different ways. In \citep{kokHS:2014,miezalTB:2016,hol:2011}, it is formulated in terms of the joint centre position, while it is written in terms of the joint centre acceleration in \citep{weygersKVVVHC:2020,leeJ:2019,Fasel18,Dorschky19}. The information can for instance be included as a constraint in an optimisation-based approach or as a measurement model in a filtering approach, \eg an extended Kalman filter. For both the formulations in terms of the joint centre position and in terms of the joint centre acceleration, it was demonstrated that exploitation of the constraint enables determining the relative orientation between the segments at least if the performed motion provides ``sufficient excitation.'' However, as elaborated above, there is no analysis or investigation on what this means and which motions render the relative motion states observable. In this work, we focus on deriving sufficient conditions for observability of the relative orientations for these arbitrary joints, i.e.\ without making assumptions on their degrees of freedom or range of motion.  

\subsection{Observability analysis for inertial motion tracking}
For any system for which one aims to build an observer, it is relevant to study the observability of the system because this gives information about whether it is indeed possible to design a stable observer for the system~\citep{besancon:2007}. For linear systems, observability can be studied by determining the rank of the well-known observability matrix~\citep{rugh:1996}. The concept of uniform complete observability is typically used to prove convergence of a Kalman filter~\citep{batistaPSO:2017}. Estimating motion using IMUs is, however, an inherently nonlinear problem~\citep{kokHS:2017}.

Nonlinear observability analysis is commonly done using Lie derivatives~\citep{besancon:2007}. This has for instance been done for vision-aided inertial navigation system~\citep{panahandehGJR:2013} and for vehicle motion estimation~\citep{marcoKS:2018}. For observability analysis for orientation estimation, \eg of kinematic chains, methods have been developed to compute these Lie derivatives on Lie groups, e.g., $\mathbb{SO}(3)$~\citep{joukovCWMPK:2019}. Alternatively, the nonlinear system can be rewritten as a linear time-varying system. This has been done for \eg inertial navigation filters~\citep{bristeauPP:2010} and for robotics applications~\citep{morin2017uniform}. In~\citep{bristeauPP:2010} it has been shown that differential observability results in uniform complete observability for linear time-variant systems and can therefore be used to prove convergence of a Kalman filter.

For nonlinear as well as for linear time-varying systems, it is known that the observability of a system may depend on the input~\citep{besancon:2007}. In the case of inertial motion tracking, observability depends on the motion of the system. In this work, we study observability for the case of magnetometer-free inertial motion tracking of a kinematic chain. We show that for this specific case, there exists an elegant way to analyse the observability of the system since it can be considered to be a special case of Wahba’s problem~\citep{wahba:1965}. For this, we neither derive Lie derivatives nor require the system to be a linear time-varying system, but instead directly study whether the relative orientation can uniquely be inferred from measurements and their derivatives, which has close connections to studying differential observability~\citep{besancon:2016,gauthierK:2001}. This allows us to systematically analyse for which motions (inputs) it is indeed possible to infer the states uniquely. We then study systematically for real-life data whether the system is indeed observable at a specific time instance.

\section{Modelling}
\label{sec:modelling}
\begin{figure}
\centering
\includegraphics[width=0.6\columnwidth]{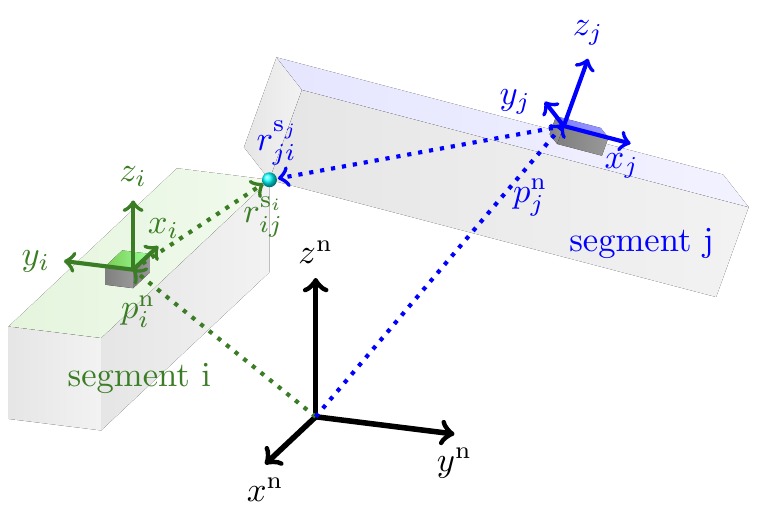}
\caption{Two adjacent segments $i$ and $j$ of a kinematic chain connected by a spherical joint.}
\label{fig:coordSys}
\end{figure}

We consider a kinematic chain with at least two rigid segments that are connected by a rotational joint, as graphically illustrated in Figure~\ref{fig:coordSys}. Our interest lies in estimating the relative pose of these connected segments, using sensors placed on each segment. In other words, our focus is on determining the relative position, velocity and orientation from the accelerometer measurements $y_{\text{a},i,t}$ and the gyroscope measurements $y_{\omega,i,t}$. Here, the subindex~$i$ explicitly indicates that these are the measurements from sensor $\text{S}_i$, where $i = 1, \hdots, N_S$. We denote the position and velocity of sensor $i$ as $p^\text{n}_i$ and $v^\text{n}_i$, respectively, where the superscript~$n$ is used to indicate that these vectors are expressed in the navigation frame $n$, which is a fixed, static coordinate frame. The origin of this frame and the direction of its axes are irrelevant in our problem formulation since we are only interested in the relative position and orientation of the sensors. The orientation of each sensor is described in terms of a rotation matrix $R^{\text{ns}_i}$, which denotes the orientation from the sensor frame $s_i$ to the static coordinate frame $n$. The origin of the frame $s_i$ lies at the centre of the accelerometer triad of sensor $i$ and its axes are aligned with the inertial sensor axes. We assume that the location and the orientation of the sensors on the body segments are known. These can for instance be obtained from pre-calibration algorithms in \citep{seelSR:2012,olssonH:2017}. Hence, estimating the relative position and orientation of the sensors becomes equivalent to estimating the relative position and orientation of the body segments. 

We assume standard measurement models for the accelerometer and gyroscope measurements
\begin{subequations}
\begin{align}
y_{\text{a},i,t} &= f_{i,t}^{\text{s}_i} + b_{\text{a},i} + e_{\text{a},i,t} 
= R^{\text{s}_i \text{n}}_t \left( a_{i,t}^\text{n} - g^\text{n} \right) + b_{\text{a},i} + e_{\text{a},i,t}, \label{eq:measModel-acc} \\
y_{\omega,i,t} &= \omega_{i,t}^{\text{s}_i} + b_{\omega,i} + e_{\omega,i,t}, \label{eq:measModel-gyr}
\end{align}
\label{eq:measModels}%
\end{subequations}
where $f_{i,t}^{\text{s}_i}$ and $\omega_{i,t}^{\text{s}_i}$, respectively, denote the specific force and the angular velocity at time $t$ of sensor $i$ expressed in sensor frame $s_i$, $a_{i,t}^\text{n}$ denotes the acceleration of sensor $i$ expressed in the navigation frame $n$ and $g^\text{n}$ denotes the Earth's gravity. Furthermore, $b_{\text{a},i}$ and $b_{\omega,i}$ denote the accelerometer and gyroscope sensor biases, respectively. These are often assumed to be constant for the duration of the data set. Finally, $e_{\text{a},i,t}$ and $e_{\omega,i,t}$ denote the accelerometer and gyroscope measurement noise, respectively, which are assumed to be white and Gaussian. 

In the remainder of this section, we will present four different continuous-time state space models that can be used to determine the relative pose of the segments. All four state space models use the inertial sensor measurements as an input to the dynamics. Hence, we express the continuous-time dynamics of the position, velocity and orientation in terms of the continuous-time specific force and angular velocity. We consider the information about the connection of the body segments as a pseudo-measurement model. The two state space models presented in Section~\ref{sec:modelling_ext} and Section~\ref{sec:modelling_red} are widely used in literature. We present two novel state space models in Section~\ref{sec:modelling_rel}. In Section~\ref{sec:mathEq} we will show that the four models are mathematically equivalent. 

\subsection{Modelling each segment's position, velocity and orientation}
\label{sec:modelling_ext}
In the model used in~\citep{kokHS:2014,miezalTB:2016}, the relative pose is estimated by parametrising the system in terms of the absolute position $p_i^\text{n}$, velocity $v_i^\text{n}$ and orientation $R_i^{\text{ns}_i}$. Assuming that the orientation is parametrised as a three-dimensional vector~\citep{kokHS:2017}, the state vector is of size $9 N_S$. Note that this is clearly an overparametrisation since the inertial sensors do not provide any information about the absolute position, velocity and heading of the body. It is, however, a very flexible model since it can straightforwardly include additional information such as GPS measurements. The dynamics is modelled as
\begin{subequations}
\begin{align}
\dot{p}_i^\text{n} &= v_i^\text{n},
\\
\dot{v}_i^\text{n} &= a_i^\text{n} = R^{\text{ns}_i} f_{i}^{\text{s}_i} + g^\text{n},
\\
\dot{R}^{\text{ns}_i} &= R^{\text{ns}_i} \left[ \omega_i^{\text{s}_i} \times \right], \label{eq:dynModel_ori}
\end{align}
\label{eq:dynModel}%
\end{subequations}
for $i = 1, \hdots , N_S$, where $[\cdot \, \times]$ denotes the matrix cross product. Note that the specific force $f_{i}$ and the angular velocity $\omega_i$ are measured by the inertial sensors as in~\eqref{eq:measModels} and are used as an input to the dynamic model, as is common in inertial sensor fusion~\citep{kokHS:2017}. The connection between the body segments is modelled as 
\begin{equation}
0 = p_i^\text{n} + R^{\text{ns}_i} r_{ij}^{\text{s}_i} - p_j^\text{n} - R^{\text{ns}_j} r_{ji}^{\text{s}_j}.
\label{eq:modelConstraint_pos}
\end{equation}
Here, $r_{ij}^{\text{s}_i}$ is the distance from sensor $S_i$ to the joint centre connecting the segments $i$ and $j$, expressed in the sensor frame $s_i$. Note the reversed subindices for $r_{ji}^{\text{s}_j}$ to denote the distance from sensor $S_j$ to the joint centre connecting the segments $i$ and $j$. Hence,~\eqref{eq:modelConstraint_pos} models the equivalence between the location of the joint centre expressed in the coordinates of sensors $i$ and $j$ and holds for any type of rotational joint. Assuming that $r_{{ij}}^{\text{s}_i}$ and $r_{ji}^{\text{s}_j}$ are known, the model can both be used as a constraint~\citep{kokHS:2014} in an optimisation problem or as a measurement model~\citep{miezalTB:2016} in \eg an extended Kalman filter implementation. 

\subsection{Modelling each segment's orientation}
\label{sec:modelling_red}
In the model used in~\citep{weygersKVVVHC:2020,leeJ:2019,Fasel18}, the relative pose is modelled only in terms of the orientation. Hence, assuming again that the orientation is parametrised as a three-dimensional vector~\citep{kokHS:2017}, the state vector is of size $3 N_S$. If the body is assumed to be rigid and the location of the sensors is known, it is possible to compute the full relative pose from the relative orientations. The dynamics is modelled by~\eqref{eq:dynModel_ori} for $i = 1, \hdots , N_S$. The connection between the body segments can again be considered to be a constraint or a measurement function and is modelled as 
\begin{align}
& 0 = R^{\text{n} \text{s}_i} \left( f_{i}^{\text{s}_i} + \left( [\omega_i^{\text{s}_i} \times]^2 + [\dot{\omega}_i^{\text{s}_i} \times] \right) r_{{ij}}^{\text{s}_i}  \right) - \nonumber \\
& \qquad \qquad R^{\text{n} \text{s}_j} \left( f_{j}^{\text{s}_j} + \left( [\omega_j^{\text{s}_j} \times]^2 + [\dot{\omega}_j^{\text{s}_j} \times] \right) r_{{ji}}^{\text{s}_j} \right),
\label{eq:modelConstraint_acc}
\end{align}
where $\dot{\omega}$ denotes the angular acceleration. This models the equivalence between the specific force or equivalently the acceleration of the joint centre expressed in the coordinates of body segments $i$ and $j$ and hence, as the model~\eqref{eq:modelConstraint_pos}, holds for any type of rotational joint. In the remainder of this paper, we will also make use of the shorthand notation
$f^{\text{s}_i}_{\text{c}_{ij},i}$ to denote the acceleration at the joint centre $c_{ij}$, measured by sensor $i$, expressed in sensor frame $s_i$, defined as
\begin{equation}
f^{\text{s}_i}_{\text{c}_{ij},i} = f_{i}^{\text{s}_i} + \left( [\omega_i^{\text{s}_i} \times]^2 + [\dot{\omega}_i^{\text{s}_i} \times] \right) r_{{ij}}^{\text{s}_i}.
\label{eq:modelConstraint_f}
\end{equation}
Using that notation,~\eqref{eq:modelConstraint_acc} can equivalently be written as 
\begin{equation}
0 = R^{\text{n} \text{s}_i} f^{\text{s}_i}_{\text{c}_{ij},i} - R^{\text{n} \text{s}_j} f^{\text{s}_j}_{\text{c}_{ij},j}.
\label{eq:modelConstraint_acc_f}
\end{equation}

\subsection{Modelling the relative pose}
\label{sec:modelling_rel}
The models from Sections~\ref{sec:modelling_ext} and~\ref{sec:modelling_red} are clearly overparametrisations of the problem. Although this is not widely used in literature, it is also possible to use a minimal representation. The benefit of overparametrisation is the flexibility of adding more information to the system. The benefit of parametrising the system with less states is that it is computationally more efficient. We define the relative position, velocity and orientation as
\begin{subequations}
\begin{align}
p_{\text{r},ij}^{\text{s}_i} &= R^{\text{s}_i \text{n}} \left( p_i^\text{n} - p_j^\text{n} \right),  \\
v_{\text{r},ij}^{\text{s}_i} &= R^{\text{s}_i \text{n}} \left( v_i^\text{n} - v_j^\text{n} \right), \\
R^{\text{s}_i \text{s}_j} &= R^{\text{s}_i \text{n}} R^{\text{n} \text{s}_j}, 
\end{align}
\end{subequations}
where $p_{\text{r},ij}^{\text{s}_i}$ and $v_{\text{r},ij}^{\text{s}_i}$ are the relative position and velocity, respectively, of segment $i$ with respect to segment $j$, expressed in segment $i$. It is now possible to write the dynamics of these relative states as
\begin{subequations}
\begin{align}
\dot{p}_{\text{r},ij}^{\text{s}_i} &= -\left[ \omega_i^{\text{s}_i} \times \right] p_{\text{r},ij}^{\text{s}_i} + v_{\text{r},ij}^{\text{s}_i}, \\
\dot{v}_{\text{r},ij}^{\text{s}_i} &= -\left[ \omega_i^{\text{s}_i} \times \right] v_{\text{r},ij}^{\text{s}_i} + f_{i}^{\text{s}_i} - R^{\text{s}_i \text{s}_j} f_{j}^{\text{s}_j}, \\
\dot{R}^{\text{s}_i \text{s}_j} &= -\left[ \omega_i^{\text{s}_i} \times \right] R^{\text{s}_i \text{s}_j} + R^{\text{s}_i \text{s}_j} \left[ \omega_j^{\text{s}_j} \times \right]. \label{eq:dynRelOri}
\end{align}
\end{subequations}
where we made use of~\eqref{eq:measModel-acc} and~\eqref{eq:dynModel}.

The model~\eqref{eq:modelConstraint_pos} can be expressed in these states as 
\begin{equation}
0 = p_{\text{r},ij}^{\text{s}_i} + r_{{ij}}^{\text{s}_i} - R^{\text{s}_i \text{s}_j} r_{{ji}}^{\text{s}_j}.
\label{eq:modelConstraint_posrel}
\end{equation}
When estimating the relative pose, the size of the state vector reduces from $9 N_S$ in Section~\ref{sec:modelling_ext} to $9 (N_S - 1)$. It is also possible to only estimate the relative orientation using~\eqref{eq:dynRelOri} and straightforwardly expressing the model~\eqref{eq:modelConstraint_acc} in terms of the relative orientation $R^{\text{s}_i \text{s}_j}$. This reduces the size of the state vector to $3 (N_S - 1)$.

\subsection{Mathematical equivalence of the four models}
\label{sec:mathEq}
Making use of the dynamic models~\eqref{eq:dynModel}, the first and second derivatives of the model~\eqref{eq:modelConstraint_pos} can be written as
\begin{subequations}
\begin{align}
&\frac{\d}{\d{t}}\left( p_i^\text{n} + R^{\text{ns}_i} r_{ij}^{\text{s}_i} - p_j^\text{n} - R^{\text{ns}_j} r_{ji}^{\text{s}_j} \right) = 
\nonumber\\
&\qquad v_i^\text{n} + R^{\text{ns}_i} \left[ \omega_i \times \right] r_{ij}^{\text{s}_i} - v_j^\text{n} - R^{\text{ns}_j} \left[ \omega_j \times \right] r_{ji}^{\text{s}_j}, \\
&\frac{\d^2}{\d t^2}\left( p_i^\text{n} + R^{\text{ns}_i} r_{ij}^{\text{s}_i} - p_j^\text{n} - R^{\text{ns}_j} r_{ji}^{\text{s}_j} \right) = 
\nonumber\\
&\qquad a_i^\text{n} + R^{\text{ns}_i} \left[ \omega_i \times \right]^2 r_{ij}^{\text{s}_i} + R^{\text{ns}_i} \left[ \dot{\omega}_i \times \right] r_{ij}^{\text{s}_i} - \nonumber \\
&\qquad \qquad a_j^\text{n} - R^{\text{ns}_j} \left[ \omega_j \times \right]^2 r_{ji}^{\text{s}_j} - R^{\text{ns}_j} \left[ \dot{\omega}_j \times \right] r_{ji}^{\text{s}_j}.
\end{align}
\label{eq:equivalence}%
\end{subequations}
It can therefore be concluded that the double derivative of the model for the joint position~\eqref{eq:modelConstraint_pos} is equal to the model for the joint acceleration~\eqref{eq:modelConstraint_acc}. In other words, these models are mathematically equivalent in the sense that an observability analysis of one of these models will straightforwardly hold for the other. Note that in practice this does not imply that the models are equivalent in every aspect. They can, for instance, behave differently in the presence of noise and, contrary to the reduced model from Section~\ref{sec:modelling_red}, the extended model from Section~\ref{sec:modelling_ext} straightforwardly opens up for including additional information such as absolute position measurements.

\section{Observability analysis}
\label{sec:observability}
In this section we will study the observability of the relative pose of the body segments using the models presented in Section~\ref{sec:modelling}. Without loss of generality we focus on a two-segment system and will use the model from Section~\ref{sec:modelling_red}. Hence, we address the question under which motions it is possible to uniquely determine the relative orientation of the two body segments using perfect (bias- and noise-free) inertial measurements and the models~\eqref{eq:dynModel_ori} and~\eqref{eq:modelConstraint_acc}. 

Let us write the model from Section~\ref{sec:modelling} as a general nonlinear state space model 
\begin{equation}
\dot{x} = f(x,u),\qquad 
y = h(x, u),
\label{eq:ssmodel}
\end{equation}
where $x \in \mathcal{M} = \mathbb{SO}^3 \times \mathbb{SO}^3$ with $x = \begin{pmatrix} x_i^\Transp & x_j^\Transp \end{pmatrix}^\Transp$ and $x_i, x_j$ denote the time-varying state vector representing the orientation of the two body segments. The (pseudo-)measurement vector $y \in \mathbb{R}^3$ is represented by~\eqref{eq:modelConstraint_acc} and models the connection between the body segments in terms of the acceleration of the joint centre. The dynamics of both segments is modelled by~\eqref{eq:dynModel_ori}. Furthermore, $u \in \mathbb{R}^{18}$ denotes the specific force, the angular velocity and the angular acceleration of each of the sensors. Note that under the assumption of bias- and noise-free measurements, the specific force and the angular velocity are directly measured by the inertial sensors. The angular acceleration can be obtained from the gyroscope measurements by numerical differentiation. 

The purpose of observability analysis is to answer the question whether the states of the model~\eqref{eq:ssmodel} can uniquely be determined using knowledge about the measurements $y$ and the inputs $u$~\citep{besancon:2007}. Assuming that the functions $f$ and $h$ in~\eqref{eq:ssmodel} are smooth functions of their arguments and that the input $u$ is smooth, we will study the observability of our system by considering $N \geq 0$ time derivatives of $y$ denoted by $y^{(N)}$. For some known input $u$, we define the mapping $\Phi_{u,N}:\,\mathcal{M}\to\mathbb{R}^{3(N+1)}$ by
\begin{equation}
\Phi_{u,N}\left(x(t)\right)=\begin{bmatrix}
y(t) \\
\dot{y}(t) \\
\vdots \\
y^{(N)}(t)
\end{bmatrix}. \label{eq:outmapp}
\end{equation}%
A system is differentially observable if, for any input $u$, there exists an $N \geq 0$ such that the mapping $\Phi_{u,N}\left(x(t)\right)$ is injective, i.e.\ there are no two points in the state space that yield the same vector of output derivatives~\citep{besancon:2016,gauthierK:2001}.

Before analysing the observability of our model~\eqref{eq:ssmodel} in detail, let us start with noticing two important properties of our model. Firstly, the observability of~\eqref{eq:ssmodel} depends on the performed motion, \ie on the inputs to the system. This can most clearly be seen by writing down the model constraint~\eqref{eq:modelConstraint_acc} and its $N$'th derivative as
\begin{subequations}
\begin{align}
y(t) &= 0 = R^{\text{n} \text{s}_i} f^{\text{s}_i}_{\text{c}_{ij},i} - R^{\text{n} \text{s}_j} f^{\text{s}_j}_{\text{c}_{ij},j}, \\
y^{(N)}(t) &= 0 = \frac{\d^{N}}{\d t^{N}}R^{\text{n} \text{s}_i} f^{\text{s}_i}_{\text{c}_{ij},i} - \frac{\d^{N}}{\d t^{N}} R^{\text{n} \text{s}_j} f^{\text{s}_j}_{\text{c}_{ij},j}, \label{eq:acc_dotacc_der}
\end{align}
\label{eq:acc_dotacc}%
\end{subequations}%
where we make use of the shorthand notation from~\eqref{eq:modelConstraint_f} and~\eqref{eq:modelConstraint_acc_f}. In the case of no movement, for instance, the time derivatives are zero and the mapping $\Phi_{u,N}$ is obviously not injective, hence the system is not observable. Because the observability depends on the inputs, we will study the observability at time $t$. 

A second important property of the system~\eqref{eq:ssmodel} is that~\eqref{eq:modelConstraint_acc} depends on both the orientation of the first and that of the second segment. The system, however, does not include any information about the absolute orientation of the sensors. It is therefore clearly not observable. Our interest, however, lies in determining the \emph{relative} orientation of the sensors. Hence, we focus only on observability of the relative orientation. In other words, we study the observability of the system~\eqref{eq:ssmodel} under the condition that one of the orientations is known or arbitrarily set to a certain value. Inspired by the concept differential observability~\citep{besancon:2016,gauthierK:2001}, we now introduce the following definition: 
\begin{definition}
The relative orientation $R^{\text{s}_i\text{s}_j}(t)$ is observable at time $t$ if, under the assumption that the orientation $R^{\text{n} \text{s}_i}(t)$ of one of the segments is known, the orientation $R^{\text{n} \text{s}_j}(t)$ of the other segment, and thus the entire state $x(t) \in \mathcal{M}$ of the system~\eqref{eq:ssmodel}, can be uniquely determined from the current input $u(t)$, the current output $y(t)$, and their time derivatives.
\label{def:relObs}
\end{definition}

Note that Definition~\ref{def:relObs} implies that it is possible to instantaneously and uniquely determine the relative orientation from the inertial measurements, the modelled connection between the body segments, and their derivatives.

Inputs for which it is not possible to uniquely determine a state $x$ from the measurements are called singular inputs~\citep{besancon:2007}. We can now analyse for which inputs the relative orientation of the system~\eqref{eq:ssmodel} is observable at time $t$ according to Definition~\ref{def:relObs}, \ie which inputs are non-singular. Note that in~\citep{bristeauPP:2010}, it has been shown that differential observability is a sufficient condition for the convergence of an observer for linear time-varying systems. However, one input being non-singular is not sufficient to reconstruct the state $x$ in the state space model~\eqref{eq:ssmodel}. A necessary condition for this is that the input is regularly persistent in the sense that there exists a $T>0$ such that, for any given time $t$, the mapping $\Phi_{u,N}\left(x(t)\right)$ is injective for at least one moment within the interval $[t,t+T]$~\citep{besancon:2016}. For the given system~\eqref{eq:ssmodel}, it can therefore be concluded that a necessary condition for a stable (nonlinear) observer design is that the relative orientation is observable in a regularly persistent manner.

We specifically focus our analysis of which inputs are non-singular only on the first order derivative in~\eqref{eq:acc_dotacc_der} since we consider the practical relevance of extending the analysis to higher orders close to zero. This brings us to the main result of our observability analysis: 

\begin{theorem} 
\label{th:obs} 
The relative orientation $R^{\text{s}_i\text{s}_j}$ of the system \eqref{eq:ssmodel} is observable according to Definition~\ref{def:relObs} for any time instant for which the specific force $a^\text{n}_{\text{c}_{ij}} - g^\text{n}$ of the joint centre is linearly independent of $\dot{a}^\text{n}_{\text{c}_{ij}}$. 
\end{theorem}

\begin{proof}
The system of equations~\eqref{eq:acc_dotacc} with $N = 1$ can be rewritten as
\begin{equation}
f^{\text{n}}_{\text{c}_{ij},i} = R^{\text{n} \text{s}_j} f^{\text{s}_j}_{\text{c}_{ij},j}, \qquad
\dot{a}^{\text{n}}_{\text{c}_{ij},i} = R^{\text{n} \text{s}_j} \left( [\omega_j^{\text{s}_j} \times] f_{c_{ij},j}^{s_j} + \dot{f}_{c_{ij},j}^{s_j} \right),
\label{eq:acc_dotacc_wahba}%
\end{equation}
where we made use of the definition of the specific force from~\eqref{eq:measModel-acc}, the time derivative of the rotation matrix according to~\eqref{eq:dynModel_ori}, the fact that the gravity vector in navigation frame is constant, and assume that the orientation $R^{\text{n} \text{s}_i}$ is known, following Definition~\ref{def:relObs}. In~\eqref{eq:acc_dotacc_wahba} we can now recognise Wahba's problem~\citep{wahba:1965}, also known as the orthogonal Procrustes problem. Based on this, it follows that the orientation $R^{\text{n} \text{s}_j}$ can uniquely be determined from~\eqref{eq:acc_dotacc_wahba} if and only if $f^{\text{n}}_{\text{c}_{ij},i}$ and $\dot{a}^{\text{n}}_{\text{c}_{ij},i}$ are linearly independent.
\end{proof}

\begin{remark} In theory, it is possible that the relative orientation is not observable at time $t$ according to Theorem~\ref{th:obs} but that the mapping $\Phi_{u,N}$ becomes injective when including higher-order derivatives of the model constraint. This consideration leads to a series of Wahba's problems that need to be checked and to the conclusion that the system is observable if any of these Wahba's problems can be solved. 
\end{remark}

In Sections~\ref{sec:simulationResults} and~\ref{sec:experimentalResults} we will present the estimation accuracy of two state estimators for different types of motion for which the relative orientation is (un)observable according to Theorem~\ref{th:obs}. We will show that accurate orientation estimates can be obtained in case the motion is observable according to Definition~\ref{def:relObs} for almost all time instants and that relatively little motion is required for observability. In case the system is not observable at any time instant according to Definition~\ref{def:relObs}, integration of the inertial measurements in~\eqref{eq:dynModel} results in a drift of the relative pose of the body segments. However, as soon as the user performs non-singular input motions, the system becomes observable and the drift can be removed.

\begin{center}
\begin{figure}
\centering
%\tikzsetnextfilename{simResultsObs}
%\tikzexternaldisable
%\setlength\figureheight{0.13\columnwidth}
%\setlength\figurewidth{0.22\columnwidth}
%\input{figures/simResultsObs.tex}
\includegraphics[]{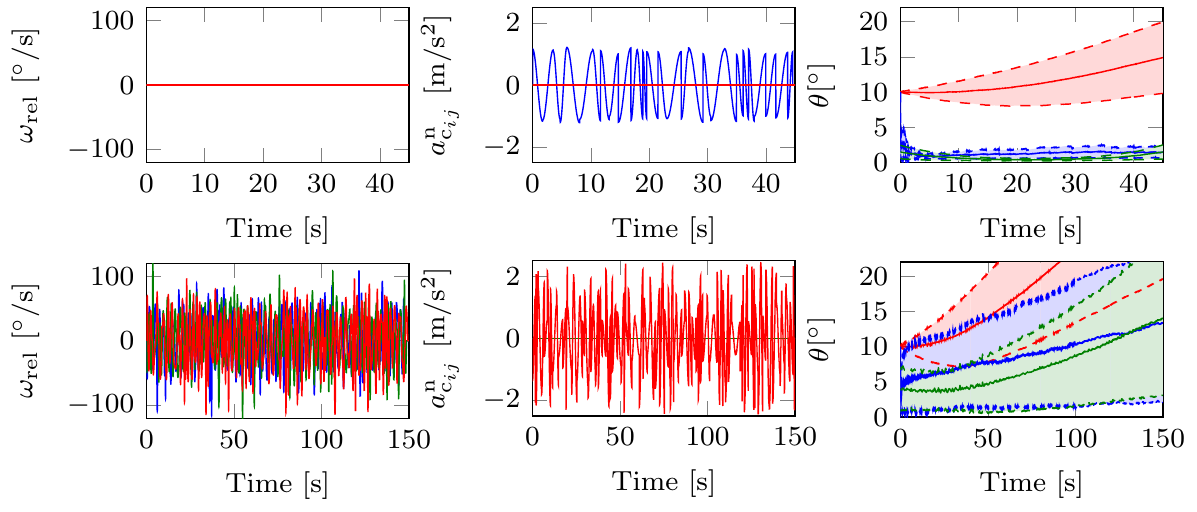}
\caption{Relative angular velocity $\omega_\text{rel}$ (left column) and acceleration of the joint centre (middle column) for the observable motion (top) and the unobservable motion (bottom) described in Section~\ref{sec:sim-obs} with the $x$-, $y$- and $z$-directions depicted in blue, green and red, respectively. Right column: Mean angular error $\theta$ and the one-standard-deviation intervals for the filtering algorithm (blue), the smoothing algorithm (green) and the integration of the gyroscope signal (red) for the observable motion (top) and the unobservable motion (bottom) described in Section~\ref{sec:sim-obs}.}
\label{fig:simResultsObs}
\end{figure}
\end{center}

\section{Simulation study}
\label{sec:simulationResults}
We illustrate the observability results from Section~\ref{sec:observability} using Monte Carlo simulations in which we simulate different types of motion. Without loss of generality we again focus on two connected segments and use the model from Section~\ref{sec:modelling_red}. To estimate the orientation of the two sensors, we use the methods presented in~\citep{weygersKVVVHC:2020}. That work presents both a filtering and a smoothing algorithm for the models~\eqref{eq:dynModel_ori} and~\eqref{eq:modelConstraint_acc} and was extensively validated on experimental data. The filtering algorithm is an extended Kalman filter, while the smoothing algorithm uses a Gauss-Newton optimisation. For more details we refer the interested reader to~\citep{weygersKVVVHC:2020} and~\citep{kokHS:2017}. 

We simulate inertial measurements at $100$Hz with noise and bias properties that have a similar order of magnitude as found in standard inertial sensors. The accelerometer noise is simulated to be zero-mean, white and Gaussian with a standard deviation of $0.05\frac{\mathrm{m}}{\mathrm{s}^2}$. The gyroscope noise is simulated to be zero-mean, white and Gaussian with a standard deviation of~$1\frac{^\circ}{\text{s}}$. Furthermore, for each Monte Carlo simulation, we draw constant accelerometer and gyroscope biases from the uniform distributions $\mathcal{U}\left(-0.05\frac{\mathrm{m}}{\mathrm{s}^2},\, 0.05\tfrac{\mathrm{m}}{\mathrm{s}^2} \right)$ and $\mathcal{U} \left(-0.2\frac{^\circ}{\mathrm{s}},\, 0.2\frac{^\circ}{\mathrm{s}} \right)$, respectively. Likewise, the distances from the sensors $i$ and $j$ to the joint are set to constant values drawn from uniform distributions, i.e. $r_{ij}^{\text{s}_i} = \begin{bmatrix} \mathcal{U}(0.01\mathrm{m}, 0.5 \mathrm{m}) & 0 & 0 \end{bmatrix}^\Transp$, $r_{ji}^{\text{s}_j} = - \begin{bmatrix} \mathcal{U} (0.01 \mathrm{m}, 0.5\mathrm{m}) & 0 & 0 \end{bmatrix}^\Transp$. We quantify the estimation error in terms of the smallest angle $\theta$ by which the estimated relative orientation $\hat{R}^{\text{s}_i \text{s}_j}$ must be rotated to become identical to the true relative orientation $R^{\text{s}_i \text{s}_j}$~\citep{Hartley:2013}. We initialise each simulation such that the initial relative orientation error is $10^\circ$. 

In Section~\ref{sec:sim-obs} we first exemplify the results from Theorem~\ref{th:obs} for both an observable and for an unobservable type of motion. Secondly, in Section~\ref{sec:excitation} we study the effect on the accuracy of the relative orientation for different amounts of excitation. Finally, we illustrate observable and unobservable motions in a specific application scenario. The results are also illustrated in the supplementary video available on \url{https://tinyurl.com/y4b2esjm}.

\subsection{Observable and unobservable motions}
\label{sec:sim-obs}
In this section we will illustrate the results from Theorem~\ref{th:obs} by simulating specific motions that are (not) observable according to the theorem and studying the estimated relative orientation using the algorithm from~\citep{weygersKVVVHC:2020}. 

First, we consider a motion that provides very little excitation of the joint system but is nevertheless non-singular in the aforementioned sense. Precisely, we assume that the specific force of the joint centre $a^\text{n}_{\text{c}_{ij}} - g^\text{n}$ has a non-zero horizontal component and that there is a non-zero change in joint centre acceleration $\dot{a}^\text{n}_{\text{c}_{ij}}$ which is not exactly in the direction of the vector $a^\text{n}_{\text{c}_{ij}} - g^\text{n}$. We simulate a specific case of this motion where the joint centre moves with non-constant acceleration for $45$ seconds along an axis that is perpendicular to the gravity direction and the two segments do not rotate around the joint centre. This motion is illustrated in terms of the relative angular velocity $\omega_\text{rel}$ and the acceleration of the joint centre $a^\text{n}_{\text{c}_{ij}}$ in Figure~\ref{fig:simResultsObs}. Note that, based on some non-specific demand for ``sufficient excitation'', one may intuitively suspect that this movement is unobservable since there is no relative motion between the two segments. However, for almost all time instants the change of acceleration is non-zero in a direction orthogonal to the acceleration due to gravity, i.e.\ the vectors $a^\text{n}_{\text{c}_{ij}} - g^\text{n}$ and $\dot{a}^\text{n}_{\text{c}_{ij}}$ are non-parallel to each other. Hence, according to Theorem~\ref{th:obs}, the relative orientation is observable according to Definition~\ref{def:relObs} for almost all time instants. The mean and standard deviation of the relative orientation estimates for the 100 Monte Carlo simulations are depicted in Figure~\ref{fig:simResultsObs}. As can be seen, the error in the relative orientation converges quickly and remains small for both the filtering and the smoothing algorithm. More specifically, after the first five seconds that are needed for the initial convergence, the mean errors of the filtering and smoothing algorithms are $1.31^\circ$ and $0.66^\circ$, respectively. The maximum errors are $4.36^\circ$ and $4.09^\circ$, respectively. For completeness, we also include the error from integration of the gyroscope only, which can be seen to grow over time up to a maximum of $26.94^\circ$.

Next, we consider motions that provide a lot of excitation but are nevertheless singular in the aforementioned sense. Specifically, we simulate data where for all time instants the joint centre acceleration changes only in the direction of the acceleration due to gravity, i.e.\ the vectors $a^\text{n}_{\text{c}_{ij}} - g^\text{n}$ and $\dot{a}^\text{n}_{\text{c}_{ij}}$ are parallel to each other. Hence, according to Theorem~\ref{th:obs}, the relative orientation is not observable according to Definition~\ref{def:relObs}. We simulate a specific case of this motion where the joint centre moves for $150$ seconds along the direction of the gravity vector with non-constant acceleration and the two segments rotate arbitrarily around the joint centre. This motion is again illustrated in Figure~\ref{fig:simResultsObs}. As can be seen, the excitation in the relative angular velocity as well as in the acceleration of the joint centre is significantly higher than in the previous example that we studied. Nevertheless, due to the unobservability of the motion, the estimation error in the relative orientation increases over time for both the filtering and the smoothing algorithm up to a maximum of $61.06^\circ$ and $45.40^\circ$, respectively, compared to a maximum error of $73.74^\circ$ for pure integration of the gyroscope signal. 

\subsection{Amount of excitation}
\label{sec:excitation}
\begin{center}
\begin{figure}
\centering
%\tikzsetnextfilename{excitationTempStatResults}
%\setlength\figureheight{0.13\columnwidth}
%\setlength\figurewidth{0.22\columnwidth}
%\input{figures/excitationTempStatResults.tex}
\includegraphics[]{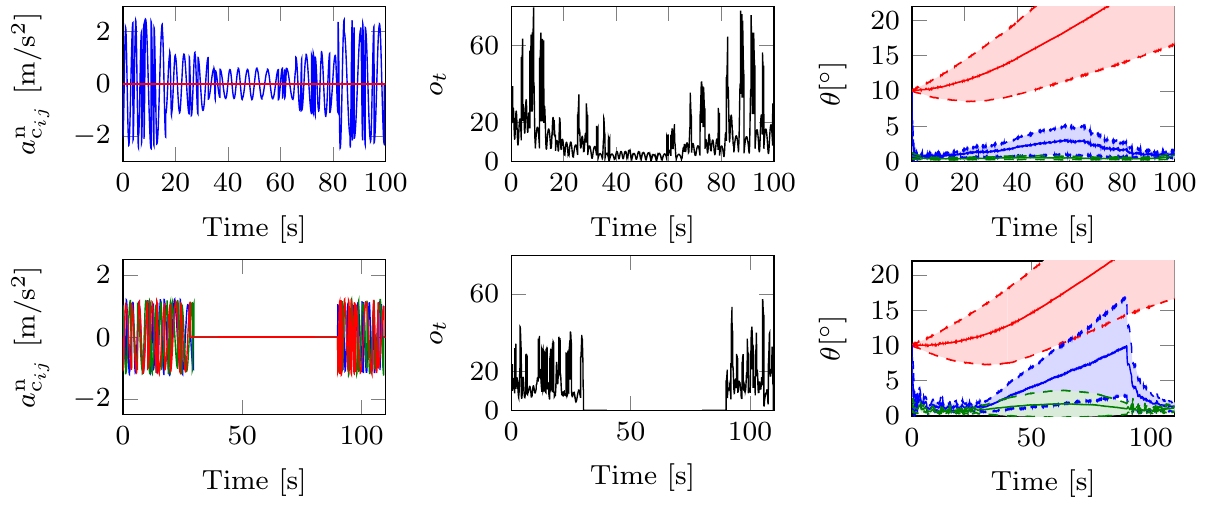}
\caption{Acceleration of the joint centre (left column), metric $o_t$ as defined in~\eqref{eq:obsIndex} (middle column) and mean angular error $\theta$ with one-standard-deviation intervals for the motion with varying amounts of excitation (top) and the temporarily stationary motion (bottom) described in Section~\ref{sec:excitation}. Left column: the $x$-, $y$- and $z$-directions are depicted in blue, green and red, respectively. Right column: The results for the filtering algorithm are depicted in blue, the results for the smoothing algorithm in green and the integration of the gyroscope signal in~red.}
\label{fig:excitationResults}
\end{figure}
\end{center}

We now study the effect of the amount of excitation on the quality of the estimates. To this end, we simulate the same motion as the observable motion in Section~\ref{sec:sim-obs} but vary the level of excitation. Hence, the relative movement of the sensors looks similar to that in the top row of Figure \ref{fig:simResultsObs}, but with the joint centre acceleration varying in excitation level as shown in Figure \ref{fig:excitationResults}. The angular error in the estimates of the filtering algorithm can be seen to increase up to a mean error of $3.02^\circ$ and a maximum error of $10.06^\circ$ during the time of low excitation. The angular errors in the estimates of the smoothing algorithm can be seen to barely be affected by this period of low excitation. To further visualise the amount of excitation and its effect on the observability of the relative pose, we compute for each time instant $t$ a metric $o_t$ that equals the norm of the cross product of $a^\text{n}_{\text{c}_{ij},t} - g^\text{n}$ and $\dot{a}^\text{n}_{\text{c}_{ij},t}$, averaged over the last 100 samples as

\begin{equation}
o_t = \tfrac{1}{100} \sum_{\tau=0}^{99} \| (a^\text{n}_{\text{c}_{ij},t-\tau} - g^\text{n}) \times \dot{a}^\text{n}_{\text{c}_{ij},t-\tau} \|_2,
\label{eq:obsIndex}
\end{equation}
where $\dot{a}^\text{n}_{\text{c}_{ij},t}$ is computed from the noise-free simulated inertial sensor data. Note that according to Theorem~\ref{th:obs}, the vectors $a^\text{n}_{\text{c}_{ij},t} - g^\text{n}$ and $\dot{a}^\text{n}_{\text{c}_{ij},t}$ need to be linearly independent for the relative orientation to be observable at time instant~$t$.

We also study the effect on the estimation accuracy of extended periods during which the metric $o_t$ is zero. To this end, we simulate $110$ seconds of data for which the joint centre arbitrarily changes its position in the reference frame with a non-constant acceleration for the first $30$ seconds and the last $20$ seconds of the simulation. During $60$ seconds in the middle of the data set, the segments are both assumed to be at rest. During these times, the relative orientation is not observable according to Definition~\ref{def:relObs}. This motion is illustrated in Figure~\ref{fig:excitationResults}. The error in the relative orientation estimated by the filtering algorithm can be seen to grow during the $60$ seconds of stationarity to a mean error of $5.47^\circ$ and a maximum error of $28.33^\circ$. In the smoothing algorithm, all measurements are used to compute the estimates. Hence, this algorithm is able to keep the errors low with a mean error of $1.44^\circ$ and a maximum error of $10.74^\circ$. This example illustrates that the filter is able to recover after temporary unobservability while estimation errors from the smoother are barely affected. Hence, for applications where longer time periods of unobservability can be expected, the use of a smoother or the use of a detector for these unobservable motions can be beneficial.

\subsection{Application to multilink aerial vehicles}
We now illustrate the role of observable and unobservable motions in the specific application of kinematically connected aerial vehicles, see~\citep{zhaoASCOI:2018} for a literature example of such a multilink robotic system. In our case, two drones perform flight manoeuvres while being rigidly connected by a ball-and-socket joint, which admits relative rotation, as illustrated in Figure~\ref{fig:quadcopters} and the supplementary video.\footnote{\url{https://tinyurl.com/y4b2esjm}} We simulate the measurement noises and biases from the same distributions as in the rest of this section, assume that the distances from the joint centres are $r_{ij}^{\text{s}_i} = \begin{bmatrix} 0.5 & 0 & 0 \end{bmatrix}^\Transp$, $r_{ji}^{\text{s}_j} = \begin{bmatrix} -0.5 & 0 & 0 \end{bmatrix}^\Transp$, and estimate the relative sensor orientations using the filtering and the smoothing method presented in~\citep{weygersKVVVHC:2020}. 

The quadcopters first take off by moving up vertically. During this motion, their relative orientation is unobservable, even when they move with respect to each other, rotate or vary their accelerations. Hence, the metric $o_t$ is zero, and the initial error of the filtering implementation does not converge. However, as soon as the joint centre starts moving horizontally, the relative orientation becomes observable, even in the absence of relative motion. We subsequently simulate a minute of data during which the quadcopters hover stationary. Hence, $o_t$ is zero and the estimation errors increase again, even for the smoothing implementation. Afterwards, when the quadcopters perform a sequence of observable motions, including horizontal displacements without any (relative) rotation, the estimation errors can be seen to converge again. Finally, during the vertical landing, the relative orientation is again unobservable. This simulation example illustrates that, in practice, long periods of unobservability or a lack of observability at the beginning of the data set can result in unreliable (real-time) estimates. However, observability can be assured by performing the motions that fulfil the proposed condition almost always. For instance, taking off or landing while simultaneously moving horizontally would yield the desired observability in the given application scenario.

\begin{figure}
\centering
\includegraphics[width=0.59\textwidth]{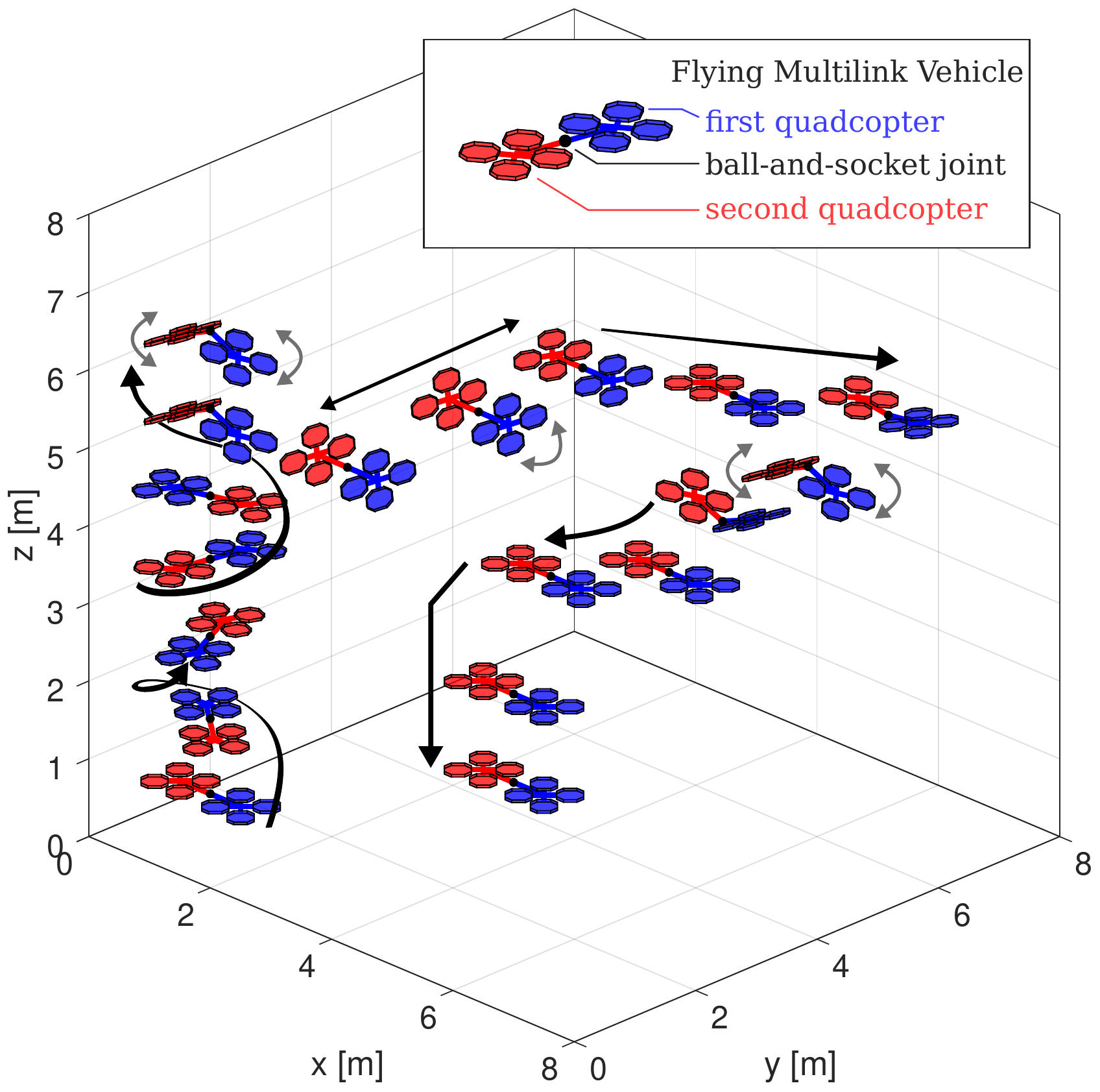} 
%\tikzsetnextfilename{quadcopter}
%\tikzexternaldisable
%\setlength\figureheight{0.168\columnwidth}
%\setlength\figurewidth{0.286\columnwidth}
%\input{figures/quadcopter.tex}
\includegraphics[]{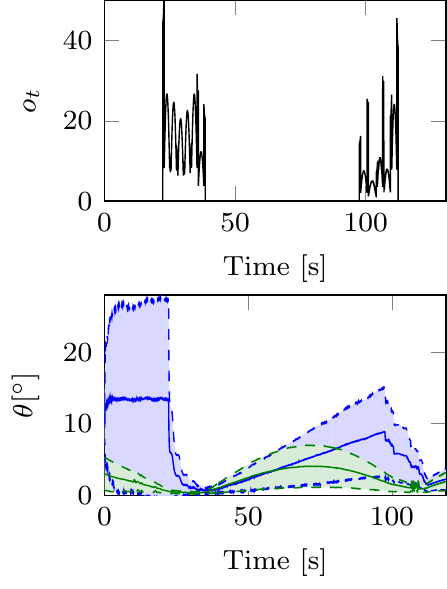} 
\caption{Left: Illustration of the simulated multilink aerial vehicle. The linked quadcopters first take off vertically and then move horizontally. Subsequently, the quadcopters remain stationary for 60 seconds, then move to different locations and finally land. Right: Metric $o_t$ and mean angular error $\theta$ with one-standard-deviation intervals. The results for the filtering algorithm are depicted in blue, the results for the smoothing algorithm in green.}
\label{fig:quadcopters}
\end{figure}
 
\section{Experiments}
\label{sec:experimentalResults}
\begin{center}
\begin{figure}
\begin{minipage}[c]{0.35\textwidth}
\centering
\vspace{-0.75cm}
    \includegraphics[width=1\columnwidth]{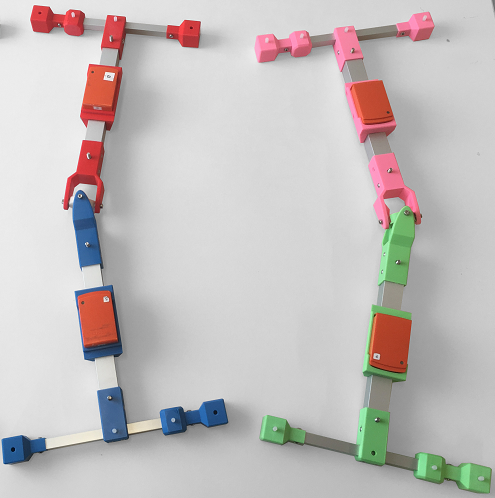}
\end{minipage}
\begin{minipage}[c]{0.65\textwidth}
\centering
%\tikzsetnextfilename{artJoints}
%\setlength\figureheight{0.2\textwidth}
%\setlength\figurewidth{0.33\textwidth}
%\input{figures/artJoints.tex}
\includegraphics[]{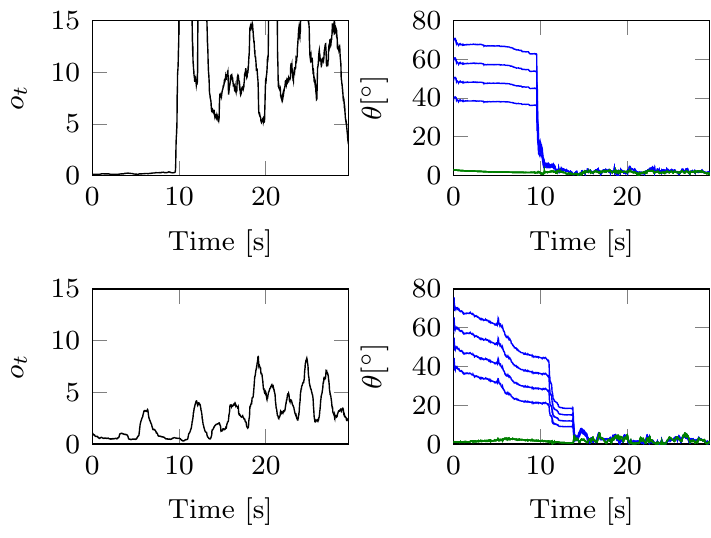}
\end{minipage}
\caption{Left: Experimental setup with 2D and 3D mechanical joints with inertial sensors attached to each segment. Right: Metric $o_t$ and the angular error $\theta$ for (top) a motion of the 3D joint where the system is first lying stationary on a table and subsequently moving around freely and for (bottom) a motion of the 2D joint where the system is held stationary but rotated twice into a different orientation, after which the system is moved around freely. The angular errors for four different initialisations are shown in blue for the filtering algorithm are depicted and in green for the smoothing algorithm.}
\label{fig:artJoints}
\end{figure}
\end{center}

\subsection{Mechanical joints}
\label{sec:mechJoints}
We experimentally validate our results with the 3D-printed mechanical joints shown in Figure~\ref{fig:artJoints}. Accelerometer and gyroscope measurements are collected using two attached IMUs (MTw Awinda, Xsens) sampled at 50 Hz. As a ground truth, marker trajectories from optical markers were simultaneously captured at 120 Hz. For different motions of the two joint types, we estimated the relative sensor orientations using the filtering and the smoothing method presented in~\citep{weygersKVVVHC:2020}. Estimated and reference relative orientations were aligned using Theorem~4.2 from~\citep{hol:2011}. The measurement noise covariances of the filtering and smoothing algorithms were chosen based on the standard deviation of the accelerometer measurements and assuming that the accelerometer measurement noise is the dominant source of error in~\eqref{eq:modelConstraint_acc}.

Firstly, we study the estimates for the 3D joint where, for the first 9.6 seconds, the system is lying stationary on a table, after which it is moved around freely, see also the supplementary video.\footnote{\url{https://tinyurl.com/y4b2esjm}} Hence, according to Theorem~\ref{th:obs}, the motion is initially unobservable, and the estimates from the filtering method from~\citep{weygersKVVVHC:2020} for four different initialisations indeed do not converge until the start of the observable motion, as shown in Figure~\ref{fig:artJoints}. The metric $o_t$ from~\eqref{eq:obsIndex} indeed shows that the motion is not observable for the first 9.6 seconds. The estimates of the smoothing algorithm can be seen to remain small for the entire data set for each of the four initialisations. To compute $o_t$, the time derivatives of the specific force $\dot{a}^\text{n}_{\text{c}_{ij},t}$ were approximated numerically based on inertial measurements of one of the two sensors and using a five-point stencil.

Secondly, we consider motion estimation for the 2D joint system, which is held mostly stationary for almost 14 seconds, but two different short rotations are performed at two distinct time instants within these 14 seconds. As can be seen in Figure~\ref{fig:artJoints}, the motion becomes observable only during these brief periods of rotation, during which the metric $o_t$ increases and the estimation error of the filtering method decreases. However, the period with observable motion is too brief for the estimation errors to decrease completely. After 14 seconds, the system is moved around freely and the estimation error can be seen to decrease further. Again, the estimates of the smoothing algorithm can be seen to remain small for the entire data set for each of the four initialisations.

\subsection{Biomechanical application}
Finally, we validate our results in a biomechanical application, where a subject consecutively walks for 198\,s in a comfortable self-selected pace and direction, sits down for 47\,s and walks again for 185\,s.
\begin{figure}
	\centering
	\includegraphics[width=0.8\columnwidth]{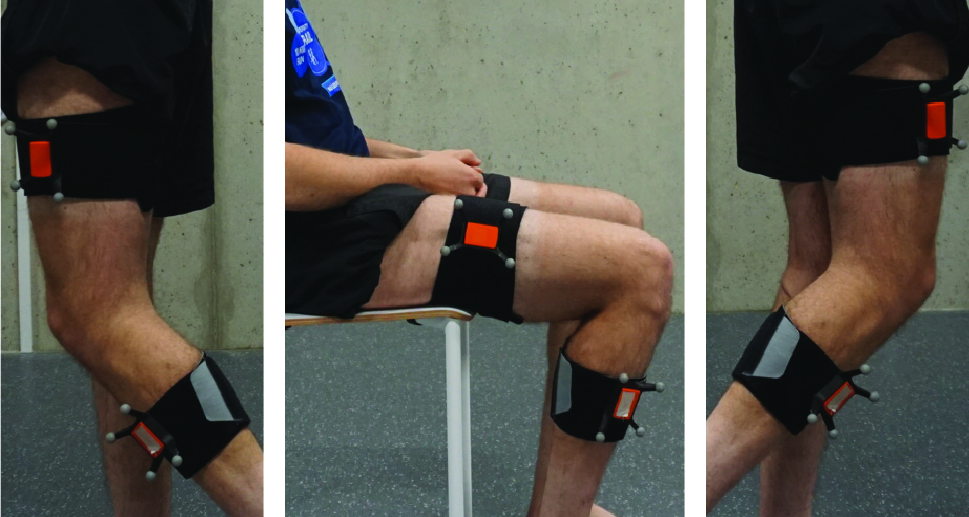}
	\caption{Experimental set-up with inertial sensors (and surrounding optical cluster markers) attached lateral and mid-distance on the thigh and shank segments. The subject executed a walk activity of 198\,s, followed by a sitting activity of 47\,s and another walk activity of 185\,s.}
	\label{fig:exp}
\end{figure}

During the experiment, accelerometer and gyroscope measurements from thigh- and shank-worn IMUs (MTw Awinda, Xsens) and marker trajectories from optical cluster markers (VICON, Vero, Vicon Motion Systems Ltd) were simultaneously captured (100Hz sampling rate) via a hardware time synchronisation. A picture of the experimental setup can be found in Figure~\ref{fig:exp}.
The study has been approved by the institutional research committee of KU Leuven (Clinical trial center UZ Leuven, Nr. S58936). All tests were done in accordance with the 1964 Helsinki declaration and its later amendments.
  
Relative sensor orientations were again estimated in a filtering and smoothing implementation following~\citep{weygersKVVVHC:2020}. Estimated and reference relative orientations were aligned using Theorem~4.2 from~\citep{hol:2011}. The resulting angular distance errors for both filtering and smoothing implementations can be found in Figure~\ref{fig:expResults}. We also show the metric $o_t$ as defined in~\eqref{eq:obsIndex}, computed similarly as in Section~\ref{sec:mechJoints}. The metric clearly shows that the relative orientation is not observable when the subject is sitting still. 
	
With respect to the knee joint, the results show that a walking movement at a self-selected comfortable pace yields sufficient excitation of the joint centre to make the relative orientation observable according to Definition~\ref{def:relObs}. The mean angular distance that the filter and smoother achieve are $5.60^\circ$ and $3.83^\circ$, respectively, for the part of the data where the person was walking. During the stationary part, the error for the smoothing implementation remained small with a maximum angular distance of $4.56^\circ$. The error of the filtering implementation grew to $19.59^\circ$ after 47 seconds, but quickly recovered at walking onset after the temporarily unobservable static time period. These observations confirm the theoretical results and are in line with the simulation results from Section~\ref{sec:excitation}.

\begin{center}
\begin{figure}
\centering
%\tikzsetnextfilename{expResults}
%\tikzexternaldisable
%\setlength\figureheight{0.2\columnwidth}
%\setlength\figurewidth{0.7\columnwidth}
%\input{figures/expResults_withObsIndex.tex}
\includegraphics[]{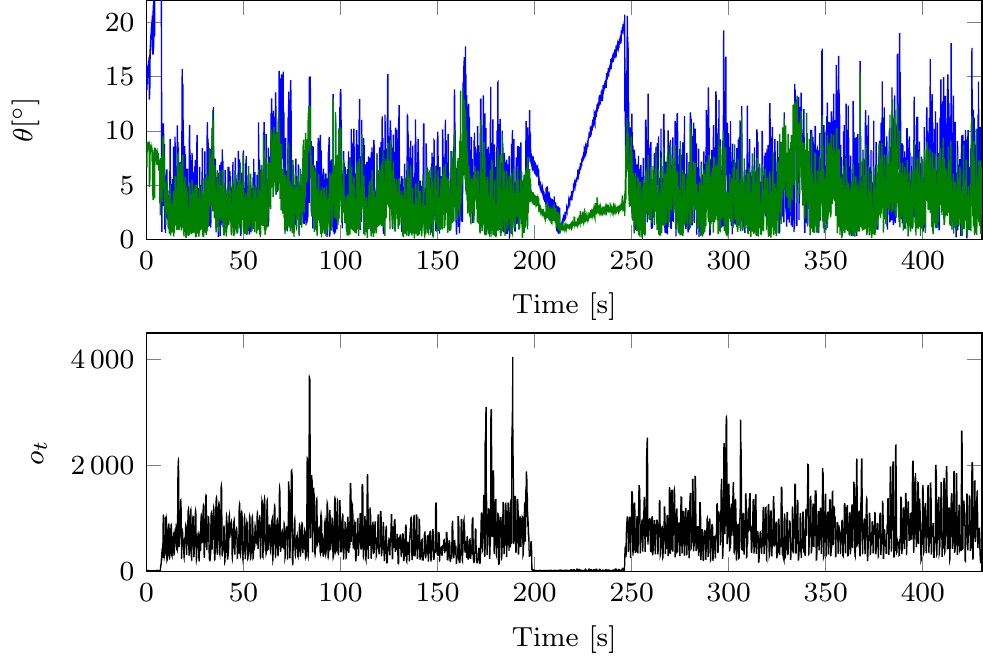}
\caption{Top: Angular error $\theta$ of the filtering algorithm (blue) and the smoothing algorithm (green) for the experimental data described in Section~\ref{sec:experimentalResults}. Bottom: The metric $o_t$, visualising the amount of excitation and its effect on the observability of the relative pose.}
\label{fig:expResults}
\end{figure}
\end{center}

\section{Conclusion}
In this work we have derived conditions on the observability of the relative pose of a kinematic chain, estimated using inertial sensors placed on adjacent segments. Any restrictive assumptions on the local magnetic field, the dynamics of the motions, or additional sensors were avoided to account for a large range of potential application scenarios. We have shown that the relative orientation is observable from purely inertial measurements whenever the specific force of the joint centre and its derivative are linearly independent. Simulations and experimental results confirmed the theoretical finding that excitation alone does not suffice to assure convergence but that accurate estimates are obtained in motions that fulfil the derived criterion. These results overcome the need to blindly hope for sufficient excitation in magnetometer-free inertial motion tracking, and they exemplify the value of systems and control theory for the design of safe and reliable sensor systems. They are expected to have an impact on a range of control applications that rely on nonrestrictive motion tracking of kinematic systems. In such applications, the derived observability condition can be used to provide crucial performance guarantees as well as to instruct users to perform motions that ensure observability. Future work could include exploiting these results in feedback-controlled systems as well as a more extensive study on the accuracy of the pose estimates as a function of the sensor noise levels and the cross product of the specific force of the joint centre and its derivative to provide bounds on the error of the estimated pose. 

\section*{Acknowledgments}
The conducted human experiments were funded by the European Regional Development Fund – We-lab for HTM~[grant number 1047]. The inertial, optical and video data of the 3D-printed mechanical joint motions has been recorded and preprocessed by Dustin Lehmann. We sincerely thank him for allowing us to use this data.

\bibliographystyle{unsrtnat}
\bibliography{references}

\end{document}